\documentclass[11pt]{article}
\usepackage{liyang}

\usepackage[usenames,divpsnames]{xcolor}

\makeatletter
\newtheorem*{rep@theorem}{\rep@title}
\newcommand{\newreptheorem}[2]{
\newenvironment{rep#1}[1]{
 \def\rep@title{#2 \ref{##1}}
 \begin{rep@theorem}\itshape}
 {\end{rep@theorem}}}
\makeatother

\newcommand{\sort}{\mathsf{Sort}}
\newcommand{\HI}{\mathsf{HI}}

\newreptheorem{theorem}{Theorem}
\newreptheorem{lemma}{Lemma}
\newreptheorem{proposition}{Proposition}
\newreptheorem{corollary}{Corollary}

\newcommand{\lnote}[1]{\footnote{\color{blue}Li-Yang : {#1}}}
\newcommand{\jnote}[1]{\footnote{\color{blue}John : {#1}}}

\renewcommand{\lnote}[1]{}
\renewcommand{\jnote}[1]{}

\newcommand{\orange}[1]{{\color{orange} #1}}
\renewcommand{\orange}[1]{#1}

\newcommand{\red}[1]{{\color{red} #1}}
\renewcommand{\red}[1]{{#1}}

\newcommand{\pCmin}{\mathsf{C}^{\oplus}_{\min}}
\newcommand{\Cmin}{\mathsf{C}_{\min}}
\newcommand{\Cert}{\mathsf{C}}
\newcommand{\pCert}{\Cert^{\oplus}} 
\newcommand{\DT}{\mathsf{DT}}
\newcommand{\PDT}{\mathsf{DT^{\oplus}}}
\newcommand{\codim}{\mathrm{codim}}
\newcommand{\sparsity}{\mathsf{sparsity}}
\newcommand{\zero}{\mathbf{0}}

\begin{document}

\title{A composition theorem for parity kill number}
\author{Ryan O'Donnell \\ Carnegie Mellon University \and Xiaorui Sun \\ Columbia University \and Li-Yang Tan \\ Columbia University \and John Wright \\ Carnegie Mellon University \and Yu Zhao \\ Carnegie Mellon University}
\maketitle

\begin{abstract}
In this work, we study the parity complexity measures $\pCmin[f]$ and $\PDT[f]$.  $\pCmin[f]$ is the \emph{parity kill number} of $f$, the fewest number of parities on the input variables one has to fix in order to ``kill" $f$, i.e.\ to make it constant.  $\PDT[f]$ is the depth of the shortest \emph{parity decision tree} which computes $f$.  These complexity measures have in recent years become increasingly important in the fields of communication complexity~\cite{ZS09, MO09, ZS10, TWXZ13} and pseudorandomness~\cite{BK12, Sha11, CT13}.

Our main result is a composition theorem for $\pCmin$.  The $k$-th power of $f$, denoted $f^{\circ k}$, is the function which results from composing $f$ with itself $k$ times.  We prove that if $f$ is not a parity function, then
\begin{equation*}
\pCmin[f^{\circ k}] \geq \Omega(\Cmin[f]^{k}).
\end{equation*}
In other words, the parity kill number of $f$ is essentially supermultiplicative in the \emph{normal} kill number of $f$ (also known as the minimum certificate complexity).

As an application of our composition theorem, we show lower bounds on the parity complexity measures of $\sort^{\circ k}$ and $\HI^{\circ k}$.  Here $\sort$ is the sort function due to Ambainis~\cite{Amb06}, and $\HI$ is Kushilevitz's hemi-icosahedron function~\cite{NW95}.  In doing so, we disprove a conjecture of Montanaro and Osborne~\cite{MO09} which had applications to communication complexity and computational learning theory.  In addition, we give new lower bounds for conjectures of~\cite{MO09,ZS10} and \cite{TWXZ13}.
\end{abstract}
\thispagestyle{empty}

\newpage 
\setcounter{page}{1}
\section{Introduction}

Recent work on the Log-Rank Conjecture has shown the importance of two related Boolean function complexity measures: sparsity and parity decision tree (PDT) depth.  The sparsity of a Boolean function, denoted $\sparsity[\hatf]$, is the number of nonzero coefficients in its Fourier transform.  A parity decision tree is a decision tree in which the nodes are allowed to query arbitrary parities of the input variables.  The PDT depth of a Boolean function, denoted $\PDT[f]$, is the depth of the shortest PDT which computes $f$.  These two quantities were linked in the papers of~\cite{MO09} and~\cite{ZS10}, both of which posed the following question:
\begin{center}
\emph{Given a sparse Boolean function, must it have a short parity decision tree?}
\end{center}
As a lower bound, any PDT computing $f$ must have depth at least $\frac{1}{2}\log(\sparsity[\hatf])$, and~\cite{MO09, ZS10} conjectured that there exists a PDT which is only polynomially worse---depth $\log(\sparsity[\hatf])^k$ for some absolute constant $k$.  Settling this question in the affirmative would prove the Log-Rank Conjecture for an important class of functions known as XOR functions (introduced in~\cite{ZS09}).  Unfortunately, at present we are very far from deciding this question.  The best known upper-bound is $\PDT[f] \leq O\left(\sqrt{\sparsity[\hatf]}\cdot \log(\sparsity[\hatf])\right)$ by~\cite{TWXZ13} (see also~\cite{STV14, Lov13}), only a square root better than the trivial $\PDT[f] \leq \sparsity[\hatf]$ bound.

A quantity intimately related to $\PDT[f]$ is the \emph{parity kill number} of a Boolean function $f$, denoted $\pCmin[f]$ (for reasons we will soon explain).  This is the fewest number of parities on the input variables one has to fix in order to ``kill" $f$, i.e.\ to make it constant.  There are several equivalent ways to reformulate this definition.  Perhaps the most familiar is in terms of \emph{parity certificate complexity}, a generalization of the ``normal'' certificate complexity measure.  Given an input $x \in \mathbb{F}_2^n$, the certificate complexity of $f$ on $x$ is the minimum number of bits $x_i$ one has to read to be certain of the value of $f(x)$.  Formally,
\begin{equation*}
\Cert[f, x] := \min\{\codim(C): C \ni x,\text{ $C$ is a subcube on which $f$ is constant}\}.
\end{equation*}
We define the minimum certificate complexity of $f$ to be $\Cmin[f] := \min_x\{\Cert[f, x]\}$.  This is the minimum number of input bits one has to fix to force $f$ to be a constant.  The parity certificate complexity of $f$ on $x$ is defined analogously, as follows:
\begin{equation*}
\pCert[f, x] := \min\{\codim(H): H \ni x\text{, $H$ is an affine subspace on which $f$ is constant}\},
\end{equation*}
and therefore $\pCmin[f] = \min_x\{\pCert[f, x]\}$.  We note here that $\Cmin[f] \geq \pCmin[f]$ always.

Given a parity decision tree $T$ for $f$, the parities that $T$ reads on input $x \in \mathbb{F}_2^n$ form a parity certificate for $x$.  As a result, $\pCmin[f]$ lower-bounds the length of any root-to-leaf path in any parity decision tree for $f$.  In particular, $\PDT[f] \geq \pCmin[f]$.  Thus, to lower-bound $\PDT[f]$, it suffices to lower-bound $\pCmin[f]$.  Remarkably, the reverse is true as well: a recent result by Tsang et al.~\cite{TWXZ13} has shown that to \emph{upper}-bound $\PDT[f]$, it suffices to \emph{upper}-bound $\pCmin[f]$\footnote{A similar argument of translating a best-case bound into a worst-case bound was recently used by Lovett in~\cite{Lov13} to show a new upper-bound for the Log-Rank Conjecture.  He showed that any total Boolean function with rank $r$ has a communication protocol of complexity $O(\sqrt{r}\cdot \log(r))$}.  More formally, they showed:
\begin{theorem}\label{thm:hongkong}
Suppose that $\pCmin[f] \leq M[f]$ for all Boolean functions $f$, where $M[f]$ is some downward non-increasing complexity measure. Then $\PDT[f] \leq M[f] \cdot \log(\sparsity[\hatf])$ for all $f$.
\end{theorem}
\noindent Here by downward non-increasing we mean that $M[f'] \leq M[f]$ whenever $f'$ can be derived from $f$ by fixing some parities on the input variables.  Theorem~\ref{thm:hongkong} implies that to prove the conjecture of~\cite{MO09, ZS10}, it suffices to show a bound of the form $\pCmin[f] \leq \log(\sparsity[\hatf])^k$, for some absolute constant $k$.  This motivates studying the properties of $\pCmin[f]$.

Another area in which parity kill number features prominently is pseudorandomness. A common scenario in this area deals with randomness extraction, in which one has access to a source that outputs mildly random bits, and the goal is to extract from these bits a set of truly random bits. A variety of tools have been developed to accomplish this goal in different settings, one of which is the \emph{affine disperser}.  An affine disperser  of dimension $d$ is simply a function $f:\F_2^n \rightarrow \F_2$ with $\pCmin[f] \geq n-d-1$.  Generally, one hopes to design dispersers with low dimension or, equivalently, a high parity kill number.  An affine disperser $f$ is ``pseudorandom'' in the sense that given inputs from a source which is supported on some large enough affine subspace $H$, $f$ will always be non-constant.  Affine dispersers have been constructed with sublinear dimension~\cite{BK12}, and the state of the art is a disperser with dimension $n^{o(1)}$~\cite{Sha11}.  The study of affine dispersers has gone hand-in-hand with studying the parity kill number of $\F_2$-polynomials; see~\cite{CT13} for an example.

\orange{Let $\DT[f]$ denote  the depth of the shortest decision tree computing $f$.}
As $\DT[f]$
is such a simple and well-understood complexity measure, one might hope to carry over intuition, and, when possible, even results, about $\DT[f]$ to the case of $\PDT[f]$. In some cases, this hope has borne fruit: an example is the following theorem from~\cite{BTW13}, which until recently was only known to hold for decision trees.
\begin{theorem}
\label{btw} 
Let $f$ be a Boolean function.  Then $\sumi \hatf(i) \le O(\PDT[f]^{1/2})$.
\end{theorem}
\noindent Another example is the OSSS inequality for decision trees~\cite{OSSS05}, which can also be shown to hold for parity decision trees by a straightforward adaptation of the proof of~\cite{JZ11}.  However, these few instances of similarity appear to be the deceptive minority rather than the majority.  On the whole, parity decision trees seem to have a much richer and more counterintuitive structure than normal decision trees, and many questions which are trivial for decision trees become interesting for parity decision trees.

\subsection{Boolean function powering}

One of the most basic operations one can perform on two Boolean functions $f:\mathbb{F}_2^n \rightarrow \mathbb{F}_2$ and $g:\mathbb{F}_2^m \rightarrow \mathbb{F}_2$, is to \emph{compose} them, producing the new function $f \circ g:\mathbb{F}_2^{m \cdot n}\rightarrow \mathbb{F}_2$.  On input $y = (y^{(1)}, \cdots, y^{(n)}) \in \left(\mathbb{F}_2^m\right)^n$,
\begin{equation*}
(f\circ g)(y) := f(g(y^{(1)}), \cdots, g(y^{(n)})).
\end{equation*}
Using this, we can construct the $k$-th \emph{power} $f^{\circ k}$ of a Boolean function recursively: $f^{\circ 1} := f$, and $f^{\circ k} := f \circ f^{\circ k-1}$.  Boolean function powering is a simple tool for generating families of Boolean functions, and it is especially useful in proving lower bounds.  It has found application in a variety of areas, from communication complexity~\cite{NW95} and Boolean function analysis~\cite{OT13} to computational learning theory~\cite{Tal13} and quantum query complexity~\cite{HLS07}.  For a comprehensive introduction to the subject of Boolean function composition and powering, see~\cite{Tal13}.

Decision tree depth is multiplicative with respect to Boolean function powering: $\DT[f^{\circ k}] = \DT[f]^k$. In addition, $\Cmin$ is supermultiplicative with respect to Boolean function powering: $\Cmin[f^{\circ k}] \geq \Cmin[f]^k$ (for simple proofs of these facts, see~\cite{Tal13}).  How might $\PDT$ and $\pCmin$ behave under powering?

Given an arbitrary Boolean function $f$, consider $f^{\circ 2} = f \circ f$.  Let us try to construct a small parity certificate for $(f \circ f)(y)$, i.e.\ a way to fix a small number of parities on the variables in $y$ to make $f \circ f$ constant.  To begin, consider a minimum (non-parity) certificate for $f(x_1, \ldots, x_n)$.  This certificate consists of a set of coordinates $\mathcal{J} \subseteq [n]$, where $|\mathcal{J}| = \Cmin[f]$, and for each $i \in \mathcal{J}$ a fixing $x_i = b_i$, for $b_i \in \mathbb{F}_2$.  The guarantee is that if each $x_i$ in $\mathcal{J}$ is set according to this certificate then $f$ is forced to be a constant.  Now we will write down a parity certificate for $f \circ f$ which, for each $i \in \mathcal{J}$, fixes $f(y^{(i)})$ to have value $b_i$.  The obvious way to do this is to separately write down the minimum parity certificate for $f(y^{(i)})$ which sets $f(y^{(i)}) = b_i$, for each $i \in \mathcal{J}$.  This gives a parity certificate for $f \circ f$ of size at least $\Cmin[f] \cdot \pCmin[f]$; we will call this the \emph{trivial certificate}. Note that if we used this process to construct a parity certificate for $f^{\circ k}$, it would have size at least $\Cmin[f]^{(k-1)} \cdot \pCmin[f]$.  In particular, the size of the trivial certificate is essentially supermultiplicative in $\Cmin[f]$.

The trivial certificate seems to only weakly use the power of parities.  
Potentially, significantly shorter certificates could exist which combine the parity certificates for the various $f(y^{(i)})$'s in clever ways.
Indeed, depending on the identity of $f$, it is sometimes possible to take small ``shortcuts'' when making the trivial certificate and save on a small number of parities.  However, these shortcuts yield parity certificates whose size is still essentially supermultiplicative in $\Cmin[f]$.  Thus, on the whole there isn't an obvious way to improve on the trivial certificate in any substantive way. It is tempting then to conjecture that $\pCmin$ is in fact supermultiplicative in $\Cmin$, and if this were true we could prove it by showing optimality of the trivial certificate.

Unfortunately, this intuition does not hold in general.  When $f$ is a parity function, $f^{\circ k}$ is also a parity function, for all $k$.  In this case, $\pCmin[f] = 1$ even though $\Cmin[f]^{(k-1)}\cdot\pCmin[f]$, the size of the trivial certificate, may be quite large.  Our main result is that if we rule out this one pathological case, then $\pCmin[f]$ is indeed supermultiplicative in $\Cmin[f]$:
\begin{theorem}\label{thm:mainresult}
Let $f:\F^n_2\to\F_2$ be a Boolean function which is not a parity.  Then
\begin{equation*}
\pCmin[f^{\circ k}] \geq \Omega(\Cmin[f]^{(k-1)}).
\end{equation*}
\end{theorem}
\noindent Note that as $\Cmin[f] \geq \pCmin[f]$, this is a stronger statement than both $\pCmin[f^{\circ k}] = \Omega(\pCmin[f]^{(k-1)})$ and $\Cmin[f^{\circ k}] = \Omega(\Cmin[f]^{(k-1)})$.  In addition, because $\PDT[f] \geq \pCmin[f]$, this shows that $\PDT[f] \geq \Omega(\Cmin[f]^{(k-1)})$.   The example of the trivial certificate shows that we cannot improve the lower bound to $\Omega(\Cmin[f]^k)$. However, as is typically the case for Boolean function powering, all that is necessary for our applications is for the exponent to be $k - o(k)$.

Most of the work in proving Theorem~\ref{thm:mainresult} comes from the special case when $\pCmin[f] \ge 2$.  The general theorem then follows from a simple reduction to this case.  For this case, we prove the following theorem:
\begin{theorem}
\label{thm:Cmin-compose}
Let $f:\F^n_2\to\F_2$ be a Boolean function with $\pCmin[f] \ge 2$.  Then
\[ \pCmin[f^{\circ k}] \ge \frac{\Cmin[f]^k-\Cmin[f]}{\Cmin[f]-1}+ \pCmin[f]= \Omega(\Cmin[f]^{(k-1)})\]
\end{theorem}
\noindent As we will see, this theorem obtains quantitatively tight bounds for certain functions $f$.

While these two theorems give a lower bound on $\PDT[f^{\circ k}]$ via the inequality $\PDT[f^{\circ k}] \geq \pCmin[f^{\circ k}]$, sometimes we can get a better lower bound if we know some additional information about $f$. In this case, we use the following theorem:
\begin{theorem}\label{thm:pdt-compose}
Let $f:\F^n_2\to \F_2$ be a Boolean function satisfying $f({\bf 0}) = 0$.  If $f$ is not a parity function, then
\[ \pCert[f^{\circ k},{\bf 0}] \ge  \Omega(\Cert[f,{\bf 0}]^{(k-1)}). \] 
In particular, we note that the LHS of the inequality is a lower bound on $\DT^{\oplus}[f]$.
\end{theorem}

\lnote{Do we want to also state here (maybe even as a formal Theorem?) that our proof can be tweaked to show the incomparable statement:
\[ \PDT[f^{\circ k}]  = \Omega(\Cert[f,{\bf 0}]^k)? \] } 
\jnote{I liked this idea, so I added it in.  Can someone check that I what I wrote for Theorem 5 is correct?} 

%
%

\subsection{Applications}

For our main application of Theorem~\ref{thm:mainresult}, we disprove one conjecture in communication complexity and show lower bounds for two related conjectures.  Let us begin by stating the conjectures.  The first we introduced above: 

\begin{conjecture}[\cite{MO09, ZS10}]\label{conj:pdt}
For every Boolean function $f$, $\PDT[f] \leq O(\log(\sparsity[\hatf])^k)$, for some absolute constant $k$.
\end{conjecture}

The next conjecture was introduced in~\cite{MO09} as a possible means of proving Conjecture~\ref{conj:pdt}.  It states, roughly, that for any Boolean function $f$, there is always a parity one can query to ``collapse'' a large part of $f$'s Fourier transform onto itself.

\begin{conjecture}[Montanaro--Osborne]\label{conj:mo}
There exists universal constants $C>0, K \in [0,1]$ such that the following holds: for every Boolean function with $\sparsity[\hatf] \ge C$ there exists $\beta \in \F^n_2$ such that
\[ \big|\supp(\hatf) \cap (\supp(\hatf)+\beta)\big| \ge K \cdot \sparsity[\hatf], \]
where $\supp(\hatf) = \{\alpha : \hatf(\alpha) \neq 0\}$, and $\supp(\hatf)+\beta = \{ \alpha + \beta \colon \alpha \in \supp(\hatf)\}$.
\end{conjecture}
If this conjecture were true, then one could construct a good parity decision tree for $f$ by always querying the parity associated with the $\beta$ guaranteed by the conjecture.  After $\log(\sparsity[\hatf])$ queries, the restricted function would have constant sparsity.  As a result, this conjecture is strong enough to imply Conjecture~\ref{conj:pdt} with $k = 1$, i.e. $\PDT[f] \leq O(\log(\sparsity[\hatf]))$. \orange{We remark that Conjecture \ref{conj:pdt} with $k=1$ also has implications outside of communication complexity: together with the inequality of Theorem \ref{btw} and the Fourier-analytic learning algorithm of \cite{OS07}, they imply an efficient algorithm for learning $\poly(n)$-sparse monotone functions from uniform random examples. This would represent a significant advance on a major open problem in learning theory, that of efficiently learning $\poly(n)$-term monotone DNF formulas.} 

The final conjecture upper bounds $\Cmin[f]$ in terms of $\lVert \hatf \rVert_1 := \sum_{\alpha} |\hatf(\alpha)|$ (this is Conjecture 27 in~\cite{TWXZ13}):
\begin{conjecture}[\cite{TWXZ13}]\label{conj:hongkong}
For every Boolean function $f$, $\pCmin[f] \leq O(\log(\lVert\hatf\rVert_1)^k)$, for some absolute constant $k$.
\end{conjecture}
Combined with Theorem~\ref{thm:hongkong}, this implies Conjecture~\ref{conj:pdt} with exponent $(k+1)$:
\begin{equation*}
\PDT[f]
\leq O(\log(\lVert\hatf\rVert_1)^k \cdot \log(\sparsity[\hatf]))
\leq O(\log(\sparsity[\hatf])^{k+1}),
\end{equation*}
where we have used here the inequality $\lVert\hatf\rVert_1 \leq \sparsity[\hatf]$.  The authors of~\cite{TWXZ13} point out that they don't know of a counterexample to Conjecture~\ref{conj:hongkong} even in the case of $k = 1$ (which was true also for Conjecture~\ref{conj:pdt}).

To prove lower bounds for these conjectures, we consider a pair of functions and the function families generated by powering them.
The first of these functions is the $\sort$ function.  This function was introduced by Ambainis in~\cite{Amb06}, in which the family of functions $\sort^{\circ k}$ was used to provide a separation between polynomial degree and quantum query complexity (see also \cite{LLS06, HLS07}).  Applying Theorem~\ref{thm:mainresult} to $\sort^{\circ k}$ yields the following corollary:
\begin{corollary}
\label{cor:sortfunction}
For infinitely many $n$,
there exists a Boolean function $f:\F^n_2\to\F_2$ satisfying
\[  \pCmin[f] = \Omega((\log(\sparsity[\hat{f}]))^{\log_2 3}) = \Omega(\log(\lVert \hatf \rVert_1)^{\log_2 3}). \]
\end{corollary}
\noindent This example shows that a lower bound of $k \geq \log_2 3 \approx 1.58$ is necessary for Conjecture~\ref{conj:hongkong}.  In fact, by using Theorem~\ref{thm:Cmin-compose}, we can exactly calculate both $\pCmin[\sort^{\circ k}]$ and $\PDT[\sort^{\circ k}]$ (see Section~\ref{sec:lowerbounds} for full details).

The second function we consider is Kushilevitz's hemi-icosahedron function $\HI$.  The family of functions $\HI^{\circ k}$ has provided the best known lower bounds for a variety of problems (e.g.\  \cite{NW95, HKP11}).  Applying Theorem~\ref{thm:pdt-compose} to $\HI^{\circ k}$ yields:
\begin{corollary}
\label{cor:HIfunction}
For infinitely many $n$,
there exists a Boolean function $f:\F^n_2\to\F_2$ satisfying
\[  \DT^\oplus[f] = \Omega((\log(\sparsity[\hat{f}]))^{\log_3 6}). \]
\end{corollary}
\noindent This example shows that a lower bound of $k \geq \log_3 6 \approx 1.63$ is necessary for Conjecture~\ref{conj:pdt}.
In addition, both Corollaries~\ref{cor:sortfunction} and~\ref{cor:HIfunction} provide examples of functions for which $\PDT[f] = \omega(\log(\sparsity[\hatf]))$, disproving Conjecture~\ref{conj:mo}.

For full details of these functions and the lower bounds, see Section~\ref{sec:lowerbounds}.  Independent of this work, Noga Ron-Zewi, Amir Shpilka, and Ben Lee Volk have also proven Corollary~\ref{cor:HIfunction} using a family of functions related to $\HI^{\circ k}$~\cite{RSV13}.
With their kind permission, we have reproduced their argument in Appendix~\ref{app:benleesproof}.

\subsection{Organization}

Section~\ref{sec:prelims} contains definitions and notations.  The most technical part of the paper is Section~\ref{sec:hardproof}, which contains the proof of Theorem~\ref{thm:Cmin-compose}.  Section~\ref{sec:easycorollaries} contains some consequences of Theorem~\ref{thm:Cmin-compose}, most importantly Theorems~\ref{thm:mainresult} and~\ref{thm:pdt-compose}. In Section~\ref{sec:lowerbounds}, we lower bound the parity complexity measures of $\sort^{\circ k}$ and $\HI^{\circ k}$, proving Corollaries~\ref{cor:sortfunction} and~\ref{cor:HIfunction}.  The alternate proof of Corollary~\ref{cor:HIfunction} by Ron-Zewi, Shpilka, and Volk can be found in Appendix~\ref{app:benleesproof}. 

%
%
%
%
%
%
%
%
%

\section{Preliminaries}\label{sec:prelims}

\subsection{Fourier analysis over the Boolean hypercube}

We will be concerned with the Fourier representation of Boolean functions and its relevant complexity measures.  In this context it will be convenient to view the output of $f$ as real numbers $-1,1\in\R$ instead of elements of $\F_2$, where we associate $0\in \F_2$ with $1\in \R$, and $1\in \F_2$ with $-1\in \R$.  Throughout this paper we will often switch freely between the two representations. \medskip

Every function $f:\F^n_2\to\R$ has a unique representation as a multilinear polynomial
\[ f(x) = \sum_{\alpha\in \F^n_2} \hatf(\alpha) \chi_\alpha(x) \quad \text{where $\chi_\alpha(x) = (-1)^{\la x,\alpha\ra}$},  \]
known as the Fourier transform of $f$.  The numbers $\hatf(\alpha)$ are the Fourier coefficients of $f$, and we refer to the $2^n$ functions $\chi_\alpha:\F^n_2\to\bits$ as the Fourier characters.  We write
$\supp(\hat{f}) = \{ \alpha \in\F^n_2 \colon \hatf(\alpha) \neq 0  \} $
to denote the support of the Fourier spectrum of $f$. 
\orange{The \emph{Fourier sparsity} of $f$, which we denote as $\sparsity[\hat{f}]$, is the cardinality of its Fourier spectrum $\supp(\hatf)$.   
\medskip

The \emph{spectral $1$-norm of $f$} is defined to be
\[ \lVert f \rVert_1 := \sum_{\alpha\in\F^n_2} |\hat{f}(\alpha)|. \]
For Boolean functions, we have  $\sparsity[\hat{f}] \geq \lVert f \rVert_1$.
}


\subsection{Parity complexity measures}
In this section, we define some relevant complexity measures.  We begin with parity decision tree complexity.
\begin{definition}[Parity decision trees]
A parity decision tree (PDT) is a binary tree where each internal node is labelled by a subset $\alpha \sse [n]$, and each leaf is labelled by a bit $b\in \F_2$.  A PDT computes a Boolean function $f:\F^n_2\to\F_2$ the natural way: on input $x\in \F^n_2$, it computes $\la x,\alpha\ra$ where $\alpha$ is the subset at the root.  If $\la x,\alpha \ra = 1$ the right subtree is recursively evaluated, and if $\la x,\alpha\ra = 0$ the left subtree is recursively evaluated.  When a leaf is reached the corresponding bit $b\in \F_2$ is the output of the function.
\end{definition}

\begin{definition}[Parity decision tree complexity]
Let $f:\F^n_2\to\F_2$ be a Boolean function.  The \emph{parity decision tree complexity of $f$}, denoted $\DT^\oplus[f]$, is the depth of the shallowest parity decision tree computing $f$.
\end{definition}

\begin{definition}[Certificate complexity]
Let $f:\F^n_2\to\F_2$ be a Boolean function.  For every $x\in\F_2^n$, the \emph{certificate complexity} and \emph{parity certificate complexity of $f$ at $x$} are defined to be
\begin{eqnarray*}
\Cert[f,x] &:=& \min\{ \codim(C) \colon \text{$C\ni x$, where $C$ is a subcube on which $f$ is constant} \}\\
 \pCert[f,x] &:=& \min\{ \codim(H)\colon \text{$H\ni x$, where $H$ an affine subspace within which $f$ is constant}\}.
 \end{eqnarray*}
The \emph{certificate complexity} and \emph{parity certificate complexity} of $f$ are
\[ \Cert[f] := \max\{ \Cert[f,x] \colon x \in \F^n_2\} \quad \text{and} \quad
\Cert^\oplus[f] := \max\{ \pCert[f,x] \colon x \in \F^n_2\}
\]
The \emph{minimum certificate complexity} and \emph{minimum parity certificate complexity} of $f$ are
\[ \Cmin[f] := \min\{ \Cert[f,x] \colon x \in \F^n_2\} \quad \text{and} \quad
\pCmin[f] := \min\{ \pCert[f,x] \colon x \in \F^n_2\}
\]
\end{definition}
The complexity measures are related as follows:
\begin{fact}
The parity complexity measures satisfy $\pCmin[f] \le \mathsf{C}^\oplus[f] \le \DT^\oplus[f]$ for every Boolean function $f$.
\end{fact}

\begin{fact}
\label{fact:cmin-supermultiplicative}
For every Boolean function $f$ and integer $k\ge 1$, we have $\Cmin[f^{\circ k}] \ge \Cmin[f]^k$.
\end{fact}
\orange{
\begin{fact}
\label{fact:cmin-zero-supermultiplicative}
For every Boolean function $f$ and integer $k\ge 1$, we have $\Cert[f^{\circ k}, {\bf 0}] \ge \Cert[f,{\bf 0}]^k$.
\end{fact}
}

Let $\calB = \{\alpha_1,\ldots,\alpha_d\} \sse \F^n_2$ be a linearly independent set of vectors, and $\sigma : \calB \to\F_2$.
We write $A[\calB,\sigma]$ to denote the affine subspace  $$A[\calB,\sigma] := \{ x\in \F_2^n \colon \la x, \alpha_i \ra = \sigma(\alpha_i) \text{ for all $1\le i \le d$}\} $$ of co-dimension $d$.   Note that $A[\calB,\sigma]$ is a linear subspace if $\sigma$ is the constant $0$ function.

\medskip

We say that coordinate $i\in [n]$ is \emph{relevant} in an affine subspace $H$ if there is an $x\in \F^n_2$ such that $x \in H$ but $x + \be_i \notin H$, and if not we say that $i$ is \emph{irrelevant}.

\begin{proposition}
\label{prop:cmin-junta}
Let $f:\F^n_2\to\F_2$ be a Boolean function and $H\sse \F^n_2$ be an affine subspace on which $f$ is constant.  Then $\Cmin[f]$ is at most the number of relevant coordinates in $H$.
\end{proposition}

\begin{proof}
Without loss of generality, suppose coordinates $i\in [k]$ are relevant in $H$ and the others are irrelevant.  Fix an arbitrary $x\in H$ and consider
\[ C = \{ y \in \F^n_2 \colon y_i = x_i \text{ for all $i\in [k]$}\}, \]
Note that $C\sse H$, since any $y\in C$ differs from $x$ only on the irrelevant coordinates of $H$.  Therefore $C$ is a subcube of co-dimension $k$ on which $f$ is constant, and so $\Cmin[f] \le \Cert[f,x] \le k$.
\end{proof}

\section{Supermultiplicativity of parity certificate complexity}\label{sec:hardproof}

In this section, we prove Theorem~\ref{thm:Cmin-compose}.

\begin{reptheorem}{thm:Cmin-compose}
Let $f:\F^n_2\to\F_2$ be a Boolean function with $\pCmin[f] \ge 2$.  Then
\[ \pCmin[f^{\circ k}] \ge \frac{\Cmin[f]^k-\Cmin[f]}{\Cmin[f]-1}+ \pCmin[f]= \Omega(\Cmin[f]^k). \]
\end{reptheorem}

Our proof uses the following strategy: given an affine subspace $H$ on which $f^{\circ k}$ is constant, we generate an affine subspace $H^*$ on which $f^{\circ (k-1)}$ is constant.  We do this by removing each $f$ on the ``outer layer'' of $f^{\circ k}$ one-by-one.  Our key step is in showing that every time we remove an $f$ on the outer layer, if that $f$ was relevant to $H$, then removing it reduces the codimension of $H$ by at least one.  This step we formalize as Proposition~\ref{prop:induction-step} below.

\begin{proposition}
\label{prop:induction-step}
Let $f^*:\F^n_2 \times \F_2\to\F_2$ and $g: \F^k_2\to\F_2$ be Boolean functions where $\pCmin[g]\ge 2$.  Define $f: \F^n_2\times \F^k_2 \to\F_2$ to be:
\[ f(x,y) = f^*(x,g(y)). \]
For any affine subspace $H\sse \F^n_2\times \F^k_2$ on which $f$ is constant, there exists an affine subspace $H^*\sse \F^n_2\times \F_2$ on which $f^*$ is constant such either:
\begin{enumerate}
\itemsep -.5pt
\item $\codim(H^*) \le \codim(H) -1$, or\label{item:DELcase}
\item the $(n+1)$-st coordinate is irrelevant in $H^*$ and $\codim(H^*) \le \codim(H)$. 
\end{enumerate}
\orange{Furthermore, among the first $n$ $x$-coordinates, any coordinate that was irrelevant in $H$ remains irrelevant in $H^*$ as well.}
\end{proposition}

\begin{proof}[Proof of Theorem \ref{thm:Cmin-compose} assuming Proposition \ref{prop:induction-step}]
Let $k\ge 2$ and consider $f^{\circ k} = f^{\circ k-1}(f,\ldots,f)$.  Let $H\sse \F^{n^k}_2$ be an affine subspace of minimum co-dimension on which $f^{\circ k}$ is constant, and so $\codim(H) = \pCmin[f^{\circ k}]$.  Applying Proposition \ref{prop:induction-step} to each of the $n^{k-1}$ base functions $f$ that $f^{\circ k-1}$ is composed with, we get an affine subspace $H^*\sse \F^{n^{k-1}}_2$ on which $f^{\circ k-1}$ is constant.  Note that the first condition of Proposition \ref{prop:induction-step} must hold at least $\Cmin[f^{\circ k-1}]$ times in this process of deriving $H^*$ from $H$, since there are at least $\Cmin[f^{\circ k-1}]$ relevant variables in $H^*$ by Proposition \ref{prop:cmin-junta}.  Therefore
\begin{eqnarray*}
\pCmin[f^{\circ k-1}] &\le& \codim(H^*) \\
&\le&  \codim(H)-\Cmin[f^{\circ k-1}] \\
&\le& \pCmin[f^{\circ k}] -\Cmin[f]^{k-1},
\end{eqnarray*}
where we have used the supermultiplicativity of $\Cmin$ (Fact \ref{fact:cmin-supermultiplicative}) for the final inequality.  Solving this recurrence completes the proof.
\end{proof}

\subsection{Proof of Proposition \ref{prop:induction-step}}

We begin with a pair of technical lemmas.  

\begin{lemma}
\label{lem:linear-restriction}
Let $g : \F^3_2 \to\F_2$.  There exists an affine subspace $H\sse \F^k_2$ of codimension at most one such that $g(x) = a_0\oplus a_1 x_1 \oplus a_2 x_1 \oplus a_3 x_3$ for all $x\in H$, where $a_0,a_1,a_2,a_3 \in \F_2$.
\end{lemma}

\begin{proof}
Since the only arity-two Boolean functions with $\F_2$-degree two are ${\sf AND}$ (two-bit conjunction) and ${\sf OR}_2$ (two-bit conjunction), we may assume that the restriction of $f$ to any subcube of co-dimension one yields either ${\sf AND}_2$ or ${\sf OR}_2$.  It follows that $f$ must be isomorphic to either
\begin{eqnarray*}
{\sf MAJ}(x_1,x_2,x_3) &=& \text{1 iff at least two input bits are 1} \\
{\sf NAE}(x_1,x_2,x_3) &=& \text{1 iff $x_1 \ne x_2$ or $x_2 \ne x_3$},
\end{eqnarray*}
both of which satisfy the lemma since they are computed by parity decision trees of depth $2$.
 \end{proof}

\begin{lemma}
\label{lem:nice-basis}
Let $H$ be an affine subspace of $\F^n_2 \times \F^k_2$.  There exists an invertible linear transformation $L = L_\ell \otimes L_r$ on $\F^n_2\times F^k_2$,  $\calB^* \sse \F^n_2\times \F^k_2$, and $\sigma^*:\calB^*\to \F_2$ such that
$A[\calB^*,\sigma] = \{ Lx \colon x \in H\}$, and $\calB^*$ can be partitioned into $\calB^* = \calB^*_x \sqcup \calB^*_y\sqcup \calB^*_{x,y}$, where
\begin{itemize}
\itemsep -.5pt
\item $\calB^*_{x,y} = \{ (\be_i,\be_i) \colon 1\le i \le t \} $
\item $\calB^*_x = \{(\be_j,\zero)  \colon t+1\le j \le t'\} $
\item $\calB^*_y = \{(\zero,\be_k) \colon t+1 \le k \le t''\},$
\end{itemize}
and $t + (t'-t) + (t''-t) = \codim(H)$.
\end{lemma}

\begin{proof}
Let $H= A[\calB,\sigma]$, where $\calB = \{ (\alpha_1,\beta_1),\ldots,(\alpha_d,\beta_d)\} \sse \F^n_2\times \F^k_2$.  First, we claim that we may assume without loss of generality that the \red{multisets} of vectors
\begin{eqnarray*}
\calB_\ell &=& \{ \alpha \in \F^n_2 - \{\zero\} \colon (\alpha,\beta) \in \calB \text{ for some $\beta\in \F^k_2$}\} \\
\calB_r &=& \{ \beta \in \F^k_2 - \{\zero\} \colon (\alpha,\beta) \in \calB \text{ for some $\alpha\in \F^n_2$}\}
\end{eqnarray*}
are each linearly independent.  Indeed, suppose there exists $\alpha_{i_1},\ldots,\alpha_{i_k} \in \calB_\ell$ such that $\alpha_{i_1} + \ldots + \alpha_{i_k} = \zero$ (an identical argument applies for $\calB_r$).   Since $\calB$ is linearly independent, there must exist some $j\in [k]$ such that $\beta_{i_j} \ne \zero$.  We note that $H$ remains the same if we replace $(\alpha_{i_j},\beta_{i_j})$ with \red{$(\zero,\beta_{i_1}+\ldots + \beta_{i_k})$, and if we set $\sigma^*(\zero, \beta_{i_1} + \ldots + \beta_{i_k}) = \sigma(\alpha_{i_1},\beta_{i_1}) + \ldots + \sigma(\alpha_{i_k},\beta_{i_k})$.  In addition, $\beta_{i_1} + \ldots + \beta_{i_k}$ can be written as a linear combination of the other elements in $\calB_r$ if and only if $\beta_{i_j}$ can.  Therefore, the number of elements in $\calB_\ell \cup \calB_r$ that can be written as a linear combination of the others  decreases by one}. Performing this replacement iteratively, the process must eventually terminate with $\calB_\ell$ and $\calB_r$ both being linearly independent.  \medskip

When $\calB_\ell$ and $\calB_r$ are linearly independent, it is straightforward to define invertible linear transformations $L_\ell$ on $\F^n_2$ mapping $\calB_\ell$ to $\{ \be_1,\ldots,\be_{|\calB_\ell|}\}$ and $L_r$ on $\F^k_2$ mapping $\calB_r$ to $\{\be_1,\ldots,\be_{|\calB_r|}\}$ accordingly, so that the invertible linear transformation $L$ on $\F^n_2\times \F^k_2$ given by by $L(x,y) = (L_\ell x, L_r y)$ maps $\calB$ into $\calB^*$ satisfying the conditions of the lemma.
\end{proof}

Now we prove Proposition~\ref{prop:induction-step}.

\begin{proof}[Proof of Proposition \ref{prop:induction-step}]
Let the input variables of $f^*:\F^n_2 \times \F_2\to\F_2$ be $x_1,\ldots,x_n\in \F^n_2$ and $z \in \F_2$, and the input variables of $g: \F^k_2\to\F_2$ be $y_1,\ldots,y_k\in\F^k_2$.
By Lemma \ref{lem:nice-basis}, we may assume that $H = A[\calB,\sigma]$ where $\calB = \calB_x \sqcup \calB_y\sqcup \calB_{x,y}$ and
\begin{itemize}
\itemsep -.5pt
\item $\calB_{x,y} = \{ (\be_i,\be_i) \colon 1\le i \le t \} $
\item $\calB_x = \{(\be_j,\zero)  \colon t+1\le j \le t'\} $
\item $\calB_y = \{(\zero,\be_k) \colon t+1 \le k \le t''\},$
\end{itemize}
and $t + (t'-t) + (t''-t) = \codim(H)$.  Let
\begin{eqnarray*}
 C_x &=& \{ x \in \F^n_2 \colon x_j = \sigma(\be_j,\zero) \text{ for all $t+1\le j\le t'$}\} \\
 C_y &=& \{ y \in \F^k_2 \colon y_k = \sigma(\zero,\be_k) \text{ for all $t+1\le k \le t''$} \} \end{eqnarray*}
be subcubes of $\F^n_2$ and $\F^k_2$ of co-dimension $|\calB_x|$ and $|\calB_y|$ respectively.  Note that $H$ comprises exactly the pairs $(x,y) \in C_x \times C_y$ satisfying $x_i \oplus y_i = \sigma(\be_i,\be_i)$ for all $1\le i \le t$.

\subsubsection{Case 1: $|\calB_y | \ge 1$ and $|\calB_{x,y}| = 0$.}

First suppose there exists $b\in \F_2$ such that $g(y) = b$ for all $y\in C_y$; by our assumption on $g$ we have $|\calB_y| \ge \pCmin[g] \ge 2$.  We claim that $f^*$ is constant on $$H^* = \{ (x,z)\colon x \in C_x \text{ and } z = b\}$$ of co-dimension $|\calB_x| + 1= (|\calB|-|\calB_y|) +1 \le |\calB| - 1.$
Indeed, suppose there exists $(x,b), (x',b) \in H^*$ such that $f^*(x,b) \neq f^*(x',b)$.  Then for any $y\in C_y$ we have $(x,y),(x',y) \in H$ and $f(x,y)\ne f(x',y)$.\medskip

   On the other hand, suppose $g$ is not constant on $C_y$.  In this case we claim that $f^*$ is constant on $H^* = \{(x,z)\colon x\in C_x\}$ of co-dimension $|\calB_x| = |\calB|-|\calB_y| \le |\calB|-1$.  Again, suppose there exists $(x,z),(x',z')\in H^*$ such that $f^*(x,z) \neq f^*(x',z')$.   Selecting $y,y'\in C_y$ such that $g(y)=z$ and $g(y') = z'$, we get $(x,y),(x',y') \in H$ such that $f(x,y) \neq f(x,y')$.

\subsubsection{Case 2: $|\calB_y| \ge 1$ and $|\calB_{x,y}| \ge 1$.}


We define subcubes $C'_x \sse C_x$ and $C'_y \sse C_y$:
 \begin{eqnarray*}
 C'_x &=& \{ x \in C_x \colon x_i = 0 \text{ for all $1\le i \le t-1$}\} \\
 C'_y &=& \{ y \in C_y \colon y_i = \sigma(\be_i,\be_j) \text{ for all $1\le i \le t-1$} \}.
 \end{eqnarray*}
Note that $C'_x$ has co-dimension $|\calB_x| + |\calB_{x,y}|-1 \le |\calB|-2$.  Furthermore, to show that a pair $(x,y)\in C'_x\times C'_y$ falls in $H$ it suffices to ensure $x_t \oplus y_t = \sigma(\be_t,\be_t)$.  We consider two possibilities: (i) there exists $a_0, a_t\in \F_2$ such that $g(y) = a_0 \oplus a_t y_t$ for all $y\in C'_y$, and otherwise (ii) there exists $b\in \F_2$ such that $g$ is non-constant on $C'_y \cap \{y\in\F^k_2\colon y_t = b\}$. \medskip

(i)  We claim that $f^*$ is constant on
\[ H^* = \{ (x,z) \colon x \in C'_x  \text{ and } z = a_0 \oplus a_t (x_t \oplus \sigma(\be_t,\be_t))  \}. \]
of co-dimension $(|\calB_x| + |\calB_{x,y}|-1)+1 \le |\calB| -1$. Indeed, suppose $f(x,z) \neq f(x',z')$ for some $(x,z),(x',z') \in H^*$.  Selecting $y,
y'\in C'_y$ such that $y_t= (x_t \oplus \sigma(\be_t,\be_t))\oplus a_0$ and $y'_t = (x'_t\oplus \sigma(\be_t,\be_t))\oplus a_0$, we get $(x,y),(x',y') \in H$ such that $f(x,y) \ne f(x',y')$.\medskip

(ii) In this case we claim that $f^*$ is constant on
\[ H^* = \{ (x,z) \colon x\in C'_x \text{ and } x_t = \sigma(\be_t,\be_t)\oplus b \}.  \]
Suppose $f(x,z)\ne f(x',z')$ for some $(x,z), (x',z') \in H^*$.   Selecting $y,y' \in C'_y \cap \{ y\in\F^k_2 \colon y_t = b\}$ satisfying $g(y) = z$ and $g(y') = z'$,  we get $(x,y),(x',y') \in H$ such that $f(x,y) \ne f(x',y')$.

\subsubsection{Case 3: $|\calB_y | = 0$ and $|\calB_{x,y}|\ge 1$.}

First suppose there exists $b_1,\ldots,b_t \in \F_2$ such that $g$ is non-constant on the subcube
$C'_y = \{ y \in \F^k_2 \colon y_i = b_i \text{ for all $1\le i \le t$}\}$.
In this case we claim that $f^*$ is constant on
\[ H^* = \{ (x,z) \colon x \in C_x  \text{ and } x_i = \sigma(\be_i,\be_i)  \oplus b_i \text{ for all $1\le i\le t$}\}. \]
Indeed, suppose there exists $(x,z),(x',z') \in H^*$ such that $f^*(x,z) \ne f^*(x',z')$.  Select $y,y'\in C'_y$ satisfying $g(y) = z$ and $g(y') =z'$, we get $(x,y),(x',y')\in H$ such that $f(x,y)\ne f(x',y')$. Note that although $\codim(H^*)$ may be as large as $|\calB|$, we have that $H^*$ is a subcube in $\F^n_2\times \F_2$ where the $(n+1)$-st coordinate is irrelevant, satisfying the second condition of the theorem statement. \medskip

Finally, if no such subcube $C'_y$ exists then $g$ is a junta over its first $t$ coordinates.  It is straightforward to verify that $t\ge 3$, since every $2$-junta has $\pCmin$ at most $1$.  Consider the sub-function $g':\F^3_2\to\F_2$ where $g'(y_1,y_2,y_3) := g(y_1,y_2,y_3,0,\ldots,0)$.   Applying Lemma \ref{lem:linear-restriction} to $g'$, we get that there exists $\alpha \in \F^3_2\times \zero^{k-3}$ and $a_0,a_1,a_2,a_3,b\in \F_2$ such that
\begin{equation}
 g'(y) = a_0 \oplus a_1y_1 \oplus a_2y_2 \oplus a_3 y_3 \text{ for all $y$ satisfying $\la y,\alpha\ra = b$}.\label{eq:linear-restriction}
 \end{equation}
Exactly two elements of $\{\be_1,\be_2,\be_3\}$ form a linearly independent set with $\alpha$.  We suppose without loss of generality that they are $\be_1$ and $\be_2$, and so $\be_3 =  \alpha+  c_1\, \be_1 + c_2\, \be_2 $ for some $c_1,c_2, \in \F_2$. \medskip

We claim that $f^*$ is constant on the affine subspace $H^*$ comprising $(x,z)\in\F^n_2\times \F_2$ satisfying all of the following conditions:
\begin{enumerate}
\itemsep -.5pt
\item[I.] $x\in C_x$.
\item[II.] $x_ i = \sigma(\be_i,\be_i)$ for all $4 \le i \le t$.
\item[III.] $x_3 = \sigma(\be_3,\be_3) \oplus  b \oplus c_1 (x_1 \oplus \sigma(\be_1,\be_1)) \oplus c_2 (x_2\oplus \sigma(\be_2,\be_2))$.
\item[IV.] $z = a_0 \oplus a_1 (x_1 \oplus \sigma(\be_1,\be_1)) \oplus a_2 (x_2 \oplus \sigma(\be_2,\be_2)) \oplus a_3 (x_3 \oplus \sigma(\be_3,\be_3))$.
\end{enumerate}
Note that $H^*$ has co-dimension $|\calB_x| + (t-3)+1 + 1 = |\calB| - 1$.  Once again, suppose $f^*(x,z) \ne f^*(x',z')$ where $(x,z),(x',z')\in H^*$.  Selecting $y\in \F^3_2 \times \zero^{k-3}$ satisfying
\begin{equation}
 y_1 = x_1\oplus \sigma(\be_1,\be_1), \quad y_2 = x_2 \oplus \sigma(\be_2,\be_2), \quad \la y, \alpha \ra = b,\label{eq:conditions-on-y}
 \end{equation}
 and likewise $y'$ for $x'$,  we claim that $(x,y),(x',y')\in H$ and $f(x,y)\ne f(x',y')$.  \medskip

 We show that $(x,y) \in H$ by checking that $x_i \oplus y_i = \sigma(\be_i,\be_i)$ for all $1\le i \le t$; the argument for $(x',y')$ is identical.  Since $y_i = 0$ for all $i \ge 4$, condition (II) of $H^*$ ensures that $x_i \oplus y_i = \sigma(\be_i,\be_i)$ for these $i$'s.  The conditions (\ref{eq:conditions-on-y}) on $y_1$ and $y_2$ above ensure that $x_i \oplus y_i = \sigma(\be_i,\be_2)$ for $i\in \{ 1,2\}$.  For $i=3$, we use the fact that
 \begin{eqnarray*} y_3 &=& \la y,\be_3 \ra \\
 &=& b \oplus c_1 y_1 \oplus c_2 y_2 \\
 &=& b \oplus c_1 (x_1 \oplus \sigma(\be_1,\be_1)) \oplus c_2 (x_2\oplus \sigma(\be_2,\be_2)),
 \end{eqnarray*}
and see that condition (III) on $H^*$ in fact ensures $x_3 \oplus y_3 = \sigma(\be_3,\be_3)$. \medskip

 To complete the proof it remains to argue that $g(y) = z$; again an identical argument establishes $g(y') = z'$.  This follows by combining (\ref{eq:linear-restriction}) and (\ref{eq:conditions-on-y}) with condition (IV) on $H^*$:
 \begin{eqnarray*}
 g(y) \ = \ g'(y) &=& a_0 \oplus a_1y_1 \oplus a_2y_2 \oplus a_3 y_3 \\
 &=& a_0 \oplus a_1 (x_1 \oplus \sigma(\be_1,\be_1)) \oplus a_2 (x_2 \oplus \sigma(\be_2,\be_2)) \oplus a_3 (x_3 \oplus \sigma(\be_3,\be_3)) \\ &=& z.
 \end{eqnarray*}
 Here the second equality is by (\ref{eq:linear-restriction}), the third by (\ref{eq:conditions-on-y}), and the final by condition (IV) on $H^*$.
 \end{proof}

\orange{ 
 \begin{remark}
 \label{remark:lb-on-PDT} 
It can be checked that in all cases, if $H$ is a linear subspace on which $f$ is constantly $0$, then $H^*$ is linear subspace on which $f^*$ is constantly $0$ as well.  Therefore, a straightforward modification of the Proof of Theorem \ref{thm:Cmin-compose} using Proposition \ref{prop:induction-step} (and Fact \ref{fact:cmin-zero-supermultiplicative}) yields the following incomparable statement: 
\begin{theorem}
\label{thm:lb-on-PDT} 
Let $f:\F^n_2\to \F_2$ be a Boolean function satisfying $f({\bf 0}) = 0$ and $\pCmin[f] \ge 2$. Then 
\[ \pCert[f^{\circ k},{\bf 0}] \ge \frac{\Cert[f,{\bf 0}]^k - \Cert[f,{\bf 0}]}{\Cert[f,{\bf 0}] - 1} + \pCert[f,{\bf 0}] = \Omega(\Cert[f,{\bf 0}]^k). \] 
In particular, we note that the LHS of the inequality is a lower bound on $\DT^{\oplus}[f]$.
\end{theorem} 
\end{remark}
}

\section{Some consequences of Theorem~\ref{thm:Cmin-compose}}\label{sec:easycorollaries}

Implicit in our proof of Theorem~\ref{thm:Cmin-compose} is the following statement:
\begin{lemma}\label{lem:composefandg}
Let $f:\F_2^n \rightarrow \F_2$ and $g:\F_2^m\rightarrow \F_2$ be Boolean functions.  If $\Cmin[g] \geq 2$, then
\begin{equation*}
\pCmin[f \circ g] \geq \pCmin[f] + \Cmin[f] \geq \Cmin[f].
\end{equation*}
\end{lemma}

Let us now derive some consequences of this statement.  First, we have the following corollary, which is almost as strong as Theorem~\ref{thm:mainresult}:
\begin{corollary}\label{cor:k-2}
Let $f:\F^n_2\to\F_2$ be a Boolean function which is not a parity.  Then
\begin{equation*}
\pCmin[f^{\circ k}] \geq \Omega(\Cmin[f]^{(k-2)}).
\end{equation*}
\end{corollary}
To prove this, we will need the following fact, which is easy to prove:
\begin{fact}\label{fact:notaparity}
Suppose $f:\F_2^n \rightarrow \F_2$ is not a parity and $\Cmin[f] \geq 2$.  Then $\pCmin[f \circ f] \geq 2$.
\end{fact}

Using this, we can prove Corollary~\ref{cor:k-2}.
\begin{proof}[Proof of Corollary~\ref{cor:k-2}.]
If $\Cmin[f] = 1$, then $\Cmin[f]^{(k-2)} = 1$ as well, and so the theorem trivially holds.  From now on, we will assume that $\Cmin[f] \geq 2$.
We may write $f^{\circ k} = f^{\circ (k-2)}\circ(f\circ f)$.  By Fact~\ref{fact:notaparity}, $\pCmin[f\circ f]\geq 2$. As a result, we can apply Lemma~\ref{lem:composefandg} to show that
\begin{equation*}
\pCmin[f^{\circ k}] = \pCmin[f^{\circ (k-2)}\circ(f\circ f)]
		\geq \Cmin[f^{\circ (k-2)}] 
		\geq \Cmin[f]^{(k-2)}.\qedhere
\end{equation*}
\end{proof}

Though Corollary~\ref{cor:k-2} is sufficient for most (if not all) applications, it is possible to slightly improve on the bound it gives using a more sophisticated argument.  At a high level, if we try using the proof of Theorem~\ref{thm:Cmin-compose} on a function $f$ for which $\pCmin[f] = 1$, then it is possible when applying Proposition~\ref{prop:induction-step} to fall into case~\ref{item:DELcase} without actually reducing the codimension of $H$ by one.  Whenever this happens, the argument essentially makes no progress, and if this always happens then there's nothing we can say about $\pCmin[f^{\circ k}]$.  Fortunately, in the case when $f$ is not a parity function, it is possible to use an amortized-analysis-style argument to show that a constant fraction of the case~\ref{item:DELcase}s \emph{do} result in reducing the codimension of $H$. This allows us to prove our main theorem, improving on Corollary~\ref{cor:k-2}.
\begin{reptheorem}{thm:mainresult}
Let $f:\F^n_2\to\F_2$ be a Boolean function which is not a parity.  Then
\begin{equation*}
\pCmin[f^{\circ k}] \geq \Omega(\Cmin[f]^{(k-1)}).
\end{equation*}
\end{reptheorem}
\noindent As the proof of this is more complicated than the proof of Corollary~\ref{cor:k-2}, we choose to omit it.

Now we have the issue of performing a similar ``bootstrapping'' on Theorem~\ref{thm:lb-on-PDT} to produce Theorem~\ref{thm:pdt-compose}.  Theorem~\ref{thm:lb-on-PDT} follows from Theorem~\ref{thm:Cmin-compose} by Remark~\ref{remark:lb-on-PDT}.  As we are just reusing the proof of Theorem~\ref{thm:Cmin-compose} to prove Theorem~\ref{thm:mainresult}, the same remark holds here.  As a result, we have the following theorem.

\begin{reptheorem}{thm:pdt-compose}
Let $f:\F^n_2\to \F_2$ be a Boolean function satisfying $f({\bf 0}) = 0$.  If $f$ is not a parity function, then
\[ \pCert[f^{\circ k},{\bf 0}] \ge  \Omega(\Cert[f,{\bf 0}]^{(k-1)}). \] 
In particular, we note that the LHS of the inequality is a lower bound on $\DT^{\oplus}[f]$.
\end{reptheorem}

We end with a remark.

\begin{remark}
Theorem~\ref{thm:mainresult} shows that $\pCmin[f^{\circ k}]$ has nontrivial exponential growth, except in the following cases:
\begin{enumerate}
\item $f$ is a parity function.\label{item:parity}
\item $\Cmin[f] = 1$, which has the following two subcases:
\begin{enumerate}
\item There exists a bit $b$ and an input $x_i$ such that $x_i = b \Rightarrow f(x) = b$.\label{item:noexponential}
\item There does \emph{not} exist a bit $b$ and an input $x_i$ such that $x_i = b \Rightarrow f(x) = b$.\label{item:exponential}
\end{enumerate}
\end{enumerate}
It is easy to see that in cases~\ref{item:parity} and~\ref{item:noexponential}, $\pCmin[f^{\circ k}] = 1$ for all $k$.  This is not so clear for case~\ref{item:exponential}, however.  In fact, we can show that in case~\ref{item:exponential}, $\pCmin[f^{\circ k}]$ has nontrivial exponential growth.  To see this, let us assume first that $k$ is even (a similar argument can be made when $k$ is odd), in which case we can write $f^{\circ k} = (f \circ f)^{k/2}$.  Now, because we're in case~\ref{item:exponential}, $\Cmin[f\circ f] \geq 2$.  Thus, we can apply Theorem~\ref{thm:mainresult} to see that $\pCmin[f^{\circ k}] \geq \Cmin[f \circ f]^{(k/2 - 1)}\geq 2^{(k/2-1)}$.  In summary, our results show that for any function $f$, either $\pCmin[f^{\circ k}] = 1$ for trivial reasons (i.e., $f$ falls in case~\ref{item:noexponential} or~\ref{item:exponential}), or $\pCmin[f^{\circ k}]$ has nontrivial exponential growth.
\end{remark}

\section{Lower bounds for specific functions}\label{sec:lowerbounds}

In this section, we show lower bounds on the parity complexity measures of $\sort^{\circ k}$ and $\HI^{\circ k}$.  Together, these prove Corollaries~\ref{cor:sortfunction} and~\ref{cor:HIfunction}.

\subsection{The $\sort$ function}

The $\sort$ function of Ambainis~\cite{Amb06} is defined as follows.
\begin{definition}
 $\sort: \F^4_2 \to \F_2$ outputs 1 if $x_1 \geq x_2 \geq x_3 \geq x_4$ or $x_1 \leq x_2 \leq x_3 \leq x_4$. Otherwise, $\sort(x_1, x_2, x_3, x_4) = 0$.
\end{definition}
Viewing $\sort$ as a function mapping $\{-1, 1\}^4 \rightarrow \{-1, 1\}$,
its Fourier expansion is the degree-$2$ homogeneous polynomial
\begin{equation}
\sort(x_1, x_2, x_3, x_4) = \frac{x_1x_2 + x_2 x_3 + x_3 x_4 - x_4 x_1}{2}.\label{eq:sort-fourier}
\end{equation}
It is easy to check that $\Cmin[\sort] = 3$, and so our Theorem~\ref{thm:mainresult} implies that $\pCmin[\sort^{\circ k}] \geq \Omega(3^k)$.  To compute the sparsity of $\sort^{\circ k}$, we first note that Equation~\ref{eq:sort-fourier} gives the recurrence
\begin{equation*}
\sparsity[\hat{\sort^{\circ k}}] = 4 \cdot \sparsity[\hat{\sort^{\circ (k-1)}}]^2.
\end{equation*}
Solving this gives $\sparsity[\hat{\sort^{\circ k}}] = 4^{2^k - 1}$.  In particular, $\log(\sparsity\big(\hat{\sort^{\circ k}}\big)\big) = O(2^k)$.  Together, these facts imply the first equality in Corollary~\ref{cor:sortfunction}.
\begin{repcorollary}{cor:sortfunction}
$\pCmin[\sort^{\circ k}] = \Omega((\log(\sparsity[\hat{\sort^{\circ k}}])^{\log_2 3}) = \Omega(\log(\lVert \hat{\sort^{\circ k}} \rVert_1)^{\log_2 3})$.
\end{repcorollary}
For the second equality, it is easy to check that every nonzero Fourier coefficient of $\sort^{\circ k}$ has equal weight (up to differences in sign).  Thus, $\lVert \hat{\sort^{\circ k}} \rVert_1 = \sqrt{\sparsity[\hat{\sort^{\circ k}}]}$, which gives the second equality.

\begin{remark}
It is also possible to verify that $\pCmin[\sort] = 2$.  Thus, the more refined bound of Theorem~\ref{thm:Cmin-compose} shows that
\[ \pCmin[\sort^{\circ k}] \ge \frac{3^k+1}{2}, \]
which is  matched \emph{exactly} by a parity decision tree for $\sort^{\circ k}$ of depth $\frac1{2}(3^k+1)$. In other words, our analysis shows that $\PDT[\sort^{\circ k}] = \pCmin[\sort^{\circ k}] = \frac1{2}(3^k+1)$, and in particular, every leaf in the optimal parity decision tree computing $\sort^{\circ k}$ has maximal depth. 
\end{remark}

\subsection{The $\HI$ function}\label{sec:HI}

\begin{definition}
The hemi-icosahedron function $\HI : \F^6_2 \rightarrow \F_2$ of Kushilevitz~\cite{NW95}  is defined as follows:  $\HI(x) = 1$ if the Hamming weight $\| x\|$ of $x$ is 1, 2 or 6, and $\HI(x) = 0$ if $\| x \|$ is 0, 4 or 5. Otherwise (i.e. $\| x \| = 3$), $\HI(x)=1$ if and only if one of the ten facets in the following diagram has all three of its vertices $1$: 
\begin{center}
\includegraphics[width=2in]{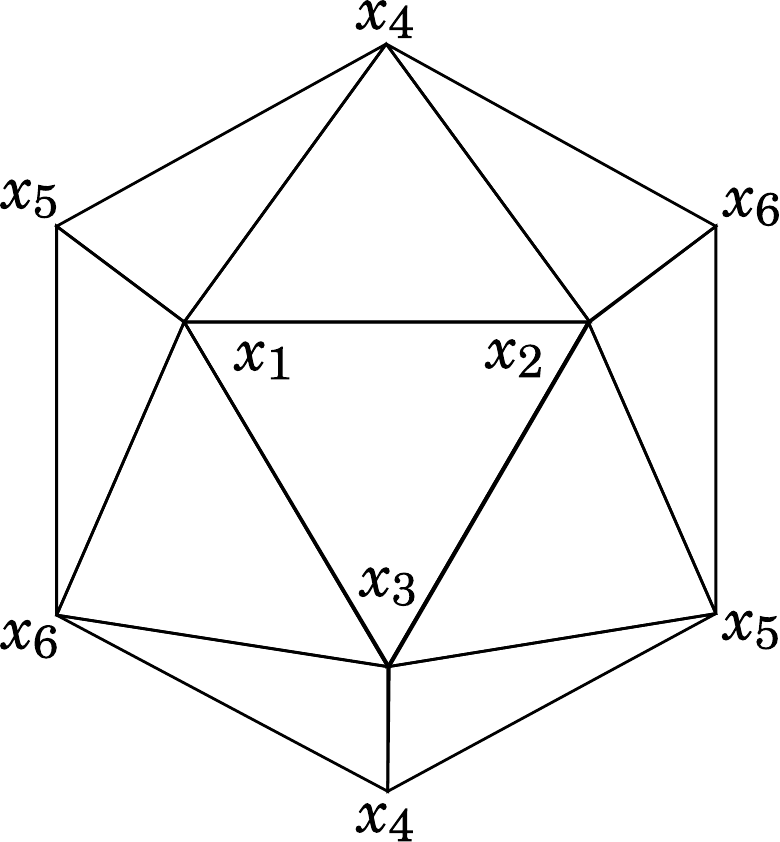}
\end{center}
\end{definition}

Viewing $\HI$ as a function mapping $\{-1, 1\}^6 \rightarrow \{-1, 1\}$,
its Fourier expansion is the degree-$3$ polynomial
\begin{align*}
\HI(x_1, \ldots, x_6)
 =  &
\frac{1}{4}\big(- \sum_i x_i  
+ x_1 x_2 x_3 + x_1 x_2 x_4 + x_1 x_3 x_6 + x_1 x_4 x_5 \\&
	+ x_1 x_5 x_6 + x_2 x_3 x_5 + x_2 x_4 x_ 6 + x_2 x_5 x_6 + x_3 x_4 x_5 + x_3 x_4 x_6
\big).
\end{align*}
Because $\HI({\bf 0}) = 0$ and $\HI(x) = 1$ for every string $x$ of Hamming weight one, $\Cert[\HI,{\bf 0}] = 6$.  As a result, our Theorem~\ref{thm:pdt-compose} implies that $\PDT[\HI^{\circ k}] \geq \Omega(6^k)$.  As for its sparsity, we refer to the following fact.
\begin{fact}\label{fact:HI}
$\sparsity\big(\hat{\HI^{\circ k}}\big) \leq 4^{3^k}$.
\end{fact}
\begin{proof}
We will first show that any Boolean function $f$ computed by a degree-$d$ polynomial has sparsity at most $4^d$.  This is true because any such polynomial is $2^{-d}$-granular, meaning that every coefficient is an integer multiple of $2^{-d}$ (this fact is exercise~$12$ in chapter~$1$ of~\cite{OD13}).  Finally, by Parseval's equation,
\begin{equation*}
1 = \sum_{\alpha\in \F_2^n} \hat{f}(\alpha)^2
	= \sum_{\alpha: \hatf(\alpha)\neq 0} \hatf(\alpha)^2
	\geq \sparsity[\hatf] \cdot \left(\frac{1}{2^d}\right)^2.
\end{equation*}
Rearranging, $\sparsity[f] \leq 4^d$.

We saw above that $\HI$ is a degree-$3$ polynomial, so $\HI^{\circ k}$ is a degree-$3^k$ polynomial.  This means that $\sparsity\big(\hat{\HI^{\circ k}}\big) \leq 4^{3^k}$.
\end{proof}
In particular, $\log(\sparsity\big(\hat{\HI^{\circ k}}\big)\big) = O(3^k)$. Putting these facts together, we get Corollary~\ref{cor:HIfunction}:

\begin{repcorollary}{cor:HIfunction}
$\DT^\oplus[\HI^{\circ k}] = \Omega((\log(\sparsity[\hat{\HI^{\circ k}}]))^{\log_3 6})$.
\end{repcorollary}

\section{Future directions}

With respect to function composition, $\DT[f]$ is a more nicely behaved complexity measure than $\Cmin[f]$.  This is because $\DT[f^{\circ k}] = \DT[f]^k$ exactly, whereas $\Cmin[f^{o k}]$ is only $\geq \Cmin[f]^k$.  On the other hand, our paper shows a composition theorem for $\pCmin[f]$ but leaves as an open problem proving a similar composition theorem for $\PDT[f]$.  We have shown that $\PDT[f]$ is supermultiplicative in $\Cmin[f]$, but it is trivial to construction functions for which $\Cmin[f]$ is small but $\PDT[f]$ is quite large.  Thus, a composition theorem for $\PDT[f]$ might prove to be useful.

\paragraph{Acknowledgments.}
We would like to thank Rocco Servedio for helpful discussions.
We also would like to thank Noga Ron-Zewi, Amir Shpilka, and Ben Lee Volk for allowing us to reproduce their argument in Appendix~\ref{app:benleesproof}.

\appendix

\section{Communication complexity proof of Corollary~\ref{cor:HIfunction}}\label{app:benleesproof}

In this section, we give the alternate proof of Corollary~\ref{cor:HIfunction} due to Ron-Zewi, Shpilka, and Volk~\cite{RSV13}.  Let $\land:\F_2^2 \rightarrow \F_2$ be the two-bit AND function.  The function family they consider is $h_k := \HI^{\circ k} \circ \land$.  Their lower bound is:

\begin{lemma}
$\PDT[h_k] = \Omega((\log(\sparsity[h_k]))^{\log_3 6})$.
\end{lemma}
\begin{proof}
Let us first calculate the sparsity of $h_k$.  As we saw in Section~\ref{sec:HI}, $\HI^{\circ k}$ is a degree-$3^k$ polynomial.  Because $\land$ is a degree-$2$ polynomial, the degree of $h_k$ is $2\cdot 3^k$.  By a similar argument as in Fact~\ref{fact:HI}, this means that $\sparsity[\hat{h_k}]\leq 4^{2\cdot 3^k}$.  In particular, $\log(\sparsity[\hat{h_k}]) \leq O(3^k)$.

Now we will show a lower bound on $\PDT[h_k]$.
The main facts that we will use about $\HI$ are that $\HI({\bf 0}) = 0$ and $\HI(x) = 1$ for every string $x$ of Hamming weight one.  These imply that $\HI^{\circ k}({\bf 0}) = 0$ and $\HI^{\circ k}(x) = 1$ for every string $x$ of Hamming weight one.

Set $n:=6^k$, the number of variables of $\HI^{\circ k}$.
Let us group the input variables of $h_k$ into two strings $x, y \in \F_2^{n}$ and write 
\begin{equation*}
h_k(x, y) = \HI^{\circ k}(x_1 \land y_1, x_2 \land y_2, \ldots, x_n \land y_n).
\end{equation*}
Consider the communication complexity scenario in which Alice is given $x$ and Bob is given $y$, and they are asked to compute $h_k(x, y)$.  If they had a parity decision tree for $h_k$ of depth~$d$, then they could compute $h_k(x, y)$ using $O(d)$ bits of communication. Define the \emph{intersection size} of $x$ and $y$ to be the number of indices $i$ for which $x_i \land y_i = 1$. It is easy to see that computing $h_k$ is equivalent to solving the Set Disjointness problem, at least when $x$ and $y$ are guaranteed to have intersection size~$0$ or~$1$ (this follows because $\HI^{\circ k}({\bf 0}) = 0$ and $\HI^{\circ k}(x) = 1$ for every string $x$ of Hamming weight one).  It is known that even in this special case, Set Disjointness requires $\Omega(n)$ bits of communication~\cite{KS92} (see also~\cite{Raz92}).  As a result, $d = \Omega(n)$, meaning that $\PDT[h_k] = \Omega(6^k)$.  Combining this with the above bound of $\log(\sparsity[\hat{h_k}]) \leq O(3^k)$ yields the lemma.
\end{proof}

\bibliographystyle{alpha}
\bibliography{wright}

\begin{thebibliography}{TWXZ13}

\bibitem[Amb06]{Amb06}
Andris Ambainis.
\newblock Polynomial degree vs.\ quantum query complexity.
\newblock {\em Journal of Computer and System Sciences}, 72(2):220--238, 2006.

\bibitem[BSK12]{BK12}
Eli Ben-Sasson and Swastik Kopparty.
\newblock Affine dispersers from subspace polynomials.
\newblock {\em SIAM Journal on Computing}, 41(4):880--914, 2012.

\bibitem[BTW13]{BTW13}
Eric Blais, Li-Yang Tan, and Andrew Wan.
\newblock Analysis of {B}oolean functions via information theory.
\newblock Manuscript, 2013.

\bibitem[CT13]{CT13}
Gil Cohen and Avishay Tal.
\newblock Two structural results for low degree polynomials and applications.
\newblock In {\em Electronic Colloquium on Computational Complexity TR13-155},
  2013.

\bibitem[HKP11]{HKP11}
Pooya Hatami, Raghav Kulkarni, and Denis Pankratov.
\newblock Variations on the sensitivity conjecture.
\newblock {\em Theory of Computing Library Graduate Surveys}, 3, 2011.

\bibitem[HLS07]{HLS07}
Peter Hoyer, Troy Lee, and Robert Spalek.
\newblock Negative weights make adversaries stronger.
\newblock In {\em Proceedings of the 39th Annual ACM Symposium on Theory of
  Computing}, pages 526--535, 2007.

\bibitem[JZ11]{JZ11}
Rahul Jain and Shengyu Zhang.
\newblock The influence lower bound via query elimination.
\newblock {\em Theory of Computing}, 7:147--153, 2011.

\bibitem[KS92]{KS92}
Bala Kalyanasundaram and Georg Schintger.
\newblock The probabilistic communication complexity of set intersection.
\newblock {\em SIAM Journal on Discrete Mathematics}, 5(4):545--557, 1992.

\bibitem[LLS06]{LLS06}
Sophie Laplante, Troy Lee, and Mario Szegedy.
\newblock The quantum adversary method and classical formula size lower bounds.
\newblock {\em Computational Complexity}, 15(2):163--196, 2006.

\bibitem[Lov13]{Lov13}
Shachar Lovett.
\newblock Communication is bounded by root of rank.
\newblock Technical report, arXiv:1306.1877, 2013.

\bibitem[MO09]{MO09}
Ashley Montanaro and Tobias Osborne.
\newblock On the communication complexity of {XOR} functions.
\newblock Technical report, arXiv:0909.3392, 2009.

\bibitem[NW95]{NW95}
Noam Nisan and Avi Wigderson.
\newblock On rank vs.\ communication complexity.
\newblock {\em Combinatorica}, 15(4):557--565, 1995.

\bibitem[O'D13]{OD13}
Ryan O'Donnell.
\newblock {\em Analysis of {B}oolean functions}.
\newblock 2013.

\bibitem[OS07]{OS07}
Ryan O'Donnell and Rocco Servedio.
\newblock Learning monotone decision trees in polynomial time.
\newblock {\em SIAM Journal on Computing}, 37(3):827--844, 2007.

\bibitem[OSSS05]{OSSS05}
Ryan O'Donnell, Michael Saks, Oded Schramm, and Rocco~A Servedio.
\newblock Every decision tree has an influential variable.
\newblock In {\em Proceedings of the 46th Annual IEEE Symposium on Foundations
  of Computer Science}, pages 31--39, 2005.

\bibitem[OT13]{OT13}
Ryan O'Donnell and Li-Yang Tan.
\newblock A composition theorem for the {F}ourier {E}ntropy-{I}nfluence
  conjecture.
\newblock In {\em Proceedings of the 40th International Colloquium on Automata,
  Languages and Programming}, pages 780--791, 2013.

\bibitem[Raz92]{Raz92}
Alexander Razborov.
\newblock On the distributional complexity of disjointness.
\newblock {\em Theoretical Computer Science}, 106(2):385--390, 1992.

\bibitem[RZSV13]{RSV13}
Noga Ron-Zewi, Amir Shpilka, and Ben~Lee Volk.
\newblock Personal communication, 2013.

\bibitem[Sha11]{Sha11}
Ronen Shaltiel.
\newblock Dispersers for affine sources with sub-polynomial entropy.
\newblock In {\em Proceedings of the 52nd Annual IEEE Symposium on Foundations
  of Computer Science}, pages 247--256, 2011.

\bibitem[STV14]{STV14}
Amir Shpilka, Avishay Tal, and Ben~Lee Volk.
\newblock On the structure of {B}oolean functions with small spectral norm.
\newblock In {\em Proceedings of the 5th Innovations in Theoretical Computer
  Science}, 2014.

\bibitem[Tal13]{Tal13}
Avishay Tal.
\newblock Properties and applications of {B}oolean function composition.
\newblock In {\em Proceedings of the 4th Innovations in Theoretical Computer
  Science}, pages 441--454, 2013.

\bibitem[TWXZ13]{TWXZ13}
Hing~Yin Tsang, Chung~Hoi Wong, Ning Xie, and Shengyu Zhang.
\newblock Fourier sparsity, spectral norm, and the log-rank conjecture.
\newblock In {\em Proceedings of the 54th Annual IEEE Symposium on Foundations
  of Computer Science}, pages 658--667, 2013.

\bibitem[ZS09]{ZS09}
Zhiqiang Zhang and Yaoyun Shi.
\newblock Communication complexities of symmetric {XOR} functions.
\newblock {\em Quantum Information and Computation}, 9(3\&4):255--263, 2009.

\bibitem[ZS10]{ZS10}
Zhiqiang Zhang and Yaoyun Shi.
\newblock On the parity complexity measures of {B}oolean functions.
\newblock {\em Theoretical Computer Science}, 411(26):2612--2618, 2010.

\end{thebibliography}

\end{document}